\documentclass{tlp}

\usepackage{graphicx}
\usepackage{times} 
\usepackage{amsmath}
\usepackage{xspace}
\usepackage{verbatim}
\usepackage{changebar}
\usepackage{floatflt}
\usepackage{ifthen}
\DeclareMathSymbol{\FORALL}   {\mathord}{symbols}{"38}
\DeclareMathSymbol{\EXISTS}   {\mathord}{symbols}{"39}
\DeclareMathSymbol{\SUCHTHAT} {\mathbin}{symbols}{"01}

\def\Exists#1#2{{\EXISTS#1}.~ #2}
\DeclareSymbolFont{AMSb}{U}{msb}{m}{n}
\DeclareMathSymbol{\N}{\mathbin}{AMSb}{"4E}
\DeclareSymbolFontAlphabet{\mathbb}{AMSb}

\newcommand{\NOTE}[1]{\ifthenelse{\boolean{draft}}{\begin{quotation}\textbf{NOTE}:#1\end{quotation}}{}}
\newcommand{\set}[1]{\ensuremath{\{#1\}}}
\newcommand{\mset}[1]{\ensuremath{\{\mathit{#1}\}}}

\newcommand{\setmin}[2]{\ensuremath{#1\!\setminus\!#2}}

\newcommand{\clos}[1]{\ensuremath{\mathit{clos}(#1)}}

\newcommand{\shoind}{$\mathcal{SHOIN}(\mathbf{D})$}

\newcommand{\shiq}{\ensuremath{\mathcal{SHIQ}}}

\newcommand{\dlrom}{\ensuremath{\mathcal{DLRO}^{-\set{\leq}}}}
\newcommand{\dlr}{\ensuremath{\mathcal{DLR}}}
\newcommand{\dlro}{\ensuremath{\mathcal{DLRO}}}

\newcommand{\shoin}{\ensuremath{\mathcal{SHOIN}}\xspace}

\newcommand{\alcnr}{\ensuremath{\mathcal{ALCNR}}}
\newcommand{\alc}{\ensuremath{\mathcal{ALC}}}
\newcommand{\al}{\ensuremath{\mathcal{AL}}}

\newcommand{\Int}{\mathcal{I}}
\newcommand{\DeltaI}{{\Delta^{\Int}}}
\newcommand{\sqs}{\sqsubseteq}

\newenvironment{program}{\[\begin{array}{rll}}{\end{array}\]}

\newcommand{\tsrule}[2]{\ensuremath{\mathit{#1} &\gets& \mathit{#2}\\}}
\newcommand{\ssrule}[2]{\ensuremath{\mathit{#1} & \gets & \mathit{#2}}}
\newcommand{\prule}[2]{\ensuremath{\mathit{#1}\gets\mathit{#2}}}
\newcommand{\pnorule}[2]{\ensuremath{{#1}\gets {#2}}}

\newcommand{\naf}[1]{not~#1}
\newcommand{\NAF}{\ensuremath{\textit{not}}}

\newcommand{\HBase}[1]{\ensuremath{\mathcal{B}_{#1}}}

\newcommand{\pred}[1]{{#1}}

\newcommand{\ground}[2]{{#1}_{#2}}

\newcommand{\preds}[1]{\ensuremath{\mathit{preds}(#1)}}

\newcommand{\vars}[1]{\ensuremath{\mathit{vars}(#1)}}

\newcommand{\dblexptime}{2-\textsc{exptime}}

\newcommand{\exptimex}[1]{{#1}-\textsc{exptime}}

\newcommand{\exptimetwonp}[1]{\ensuremath{2\mbox{-}\textsc{exptime}^{\textsc{np}}}}
\newcommand{\exptimetwonexptime}[1]{\ensuremath{2\mbox{-}\textsc{exptime}^{\textsc{nexptime}}}}
\newcommand{\exptimetwonexptimetwo}[1]{\ensuremath{2\mbox{-}\textsc{exptime}^{2\mbox{-}\textsc{nexptime}}}}

\newcommand{\dlrrole}{\ensuremath{(\$ i / n:C)}}
\newcommand{\dlrrolei}[2]{\ensuremath{(\$ #1 / n:#2)}}
\newcommand{\dlrrolein}[3]{\ensuremath{(\$ #1 / #2:#3)}}
\newcommand{\dlrconcept}{\ensuremath{\exists[\$ i]\mathbf{R}}}
\newcommand{\dlrconcepti}[2]{\ensuremath{\exists[\$ #1]#2}}
\newcommand{\dlrconceptless}{\ensuremath{\leq\! k [\$ i]\mathbf{R}}}

\newcommand{\card}[1]{\ensuremath{\mbox{\ensuremath{\vert #1 \vert}}}}

\newcommand{\posi}[1]{\ensuremath{{#1}^{+}}}
\newcommand{\nega}[1]{\ensuremath{{#1}^{-}}}

\newcommand{\lit}[1]{\ensuremath{\mathit{#1}}}

\newcommand{\taxiom}[2]{\ensuremath{\mathit{#1} &\sqs &\mathit{#2}\\}}

\newenvironment{knowb}{\[\begin{array}{rll}}{\end{array}\]}

\newcommand{\cts}[1]{\ensuremath{\mathit{cts}{(#1)}}}

\newcommand{\dllog}{\ensuremath{\mathcal{DL}\mathit{+log}}}
\newcommand{\dl}{\ensuremath{\mathcal{DL}}}

\newtheorem{definition}{Definition}
\newtheorem{example}{Example}
\newtheorem{theorem}{Theorem}
\newtheorem{corollary}{Corollary}

 \submitted{20 June 2006}
 \revised{3 January 2007}
 \accepted{18 October 2007}

\begin{document}
\title[G-Hybrid Knowledge Bases]{Guarded Hybrid Knowledge Bases\footnote{A preliminary version of this paper appeared in the proceedings
of the  \emph{ICLP'06 Workshop on Applications of Logic Programming in the Semantic Web
and Semantic Web Services (ALPSWS2006)} pages 39-54, Seattle, Washington, USA, August 16 2006.}}

\author[S.~Heymans et al.]{STIJN HEYMANS$^1$, JOS DE BRUIJN$^2$, LIVIA PREDOIU$^3$, CRISTINA FEIER$^1$,
\authorbreak  DAVY VAN NIEUWENBORGH$^4$ \thanks{
The work is funded by the European Commission under the projects
ASG, DIP, enIRaF, InfraWebs, Knowledge Web, Musing, Salero, SEKT,
SEEMP, SemanticGOV, Super, SWING and TripCom; by Science Foundation
Ireland under the DERI-Lion Grant No.SFI/02/CE1/I13 ; by the FFG
(\"Osterreichische Forschungsförderungsgeselleschaft mbH) under the
projects Grisino, RW$^2$, SemNetMan, SEnSE, TSC and OnTourism.
Davy Van Nieuwenborgh was supported by the Flemish Fund
for Scientific Research (FWO-Vlaanderen).}\\
$^1$ Digital Enterprise Research Institute,
University of Innsbruck, Technikerstrasse 21a, Innsbruck, Austria \\
\email{\{stijn.heymans,cristina.feier\}@deri.at} \\
$^2$ Faculty of Computer Science, Free University of Bozen-Bolzano, Italy \\
\email{debruijn@inf.unibz.it} \\
$^3$ Institute of Computer Science,
    University of Mannheim,
     A5, 6 68159 Mannheim,
             Germany \\
             \email{livia@informatik.uni-mannheim.de}\\
$^4$ Dept. of Computer Science,
Vrije Universiteit Brussel, VUB,
Pleinlaan 2, B1050 Brussels, Belgium \\
\email{dvnieuwe@vub.ac.be}
    }

\maketitle
\begin{abstract}{Recently, there has been a lot of interest in the
    integration of Description Logics and rules on the Semantic Web.}
    We define \emph{guarded hybrid knowledge bases} (or \emph{g-hybrid knowledge bases}) as knowledge bases that
    consist of a Description Logic knowledge base and a \emph{guarded}
    logic program, similar to the \dllog{} knowledge bases from
    \cite{Rosa06b}. G-hybrid knowledge bases enable an integration of
    Description Logics and Logic Programming where, unlike in other
    approaches, variables in the rules of a guarded program do not
    need to appear in positive non-DL atoms of the body, i.e.~DL atoms can
    act as \emph{guards} as well. Decidability of satisfiability
    checking of g-hybrid knowledge bases is shown for the particular
    DL \dlrom{}{, which is close to OWL DL,} by a reduction to guarded programs under
    the open answer set semantics.  Moreover, we show
    \dblexptime-completeness for satisfiability checking of such g-hybrid knowledge bases.
    Finally, we discuss advantages
    and disadvantages of our approach compared with \dllog{} knowledge
    bases.
\end{abstract}


\section{Introduction}\label{sec:intro}

The integration of
Description Logics with rules has recently received a lot of attention in the context of the Semantic Web
\cite{rosati:05,Rosa06b,eiter2004,motik,Horrocks+Patel-Schneider-PropRuleLang:04,motik07,bruijn07}. R-hybrid knowledge
bases \cite{rosati:05}, and its extension \dllog{} \cite{Rosa06b}, is
an elegant formalism based on combined models for Description Logic
knowledge bases and nonmonotonic logic programs. We propose a variant
of r-hybrid knowledge bases, called \emph{g-hybrid knowledge bases},
which do not require standard names or a special safeness restriction on rules, but instead require the program to be \emph{guarded}.
We show several computational properties by a reduction to
guarded open answer set programming \cite{hvnv-lpnmr2005,Heymans+NieuwenborghETAL-OpenAnswProgwith:06}.

Open answer set programming (OASP) \cite{hvnv-lpnmr2005,Heymans+NieuwenborghETAL-OpenAnswProgwith:06}
combines the logic programming and first-order logic paradigms.  From
the logic programming paradigm it inherits a rule-based presentation
and a nonmonotonic semantics by means of negation as failure. In
contrast with usual logic programming semantics, such as the answer
set semantics \cite{gelfond88stable}, OASP allows for domains
consisting of other objects than those present in the logic program at
hand.  Such open domains are inspired by first-order logic based
languages such as Description Logics (DLs) \cite{dlbook} and make OASP
a viable candidate for conceptual reasoning.  Due to its rule-based
presentation and its support for nonmonotonic reasoning and open
domains, OASP can be used to reason with both rule-based and
conceptual knowledge on the Semantic Web, as illustrated in
\cite{hvnv:eswc2005}.

A major challenge for OASP is to control undecidability of
satisfiability checking, a challenge it shares with DL-based
languages.  In \cite{hvnv-lpnmr2005,Heymans+NieuwenborghETAL-OpenAnswProgwith:06},
we identify a decidable class of programs,
the so-called \emph{guarded programs}, for which decidability
of satisfiability checking is obtained by a translation to guarded
fixed point logic \cite{gradel}.  In \cite{hvnfv-row2006}, we show the
expressiveness of such guarded programs by simulating a DL with
$n$-ary roles and nominals.  In particular, we extend the DL \dlr{}
\cite{DL-97} with both \emph{concept nominals} $\set{o}$ and
\emph{role nominals} $\set{(o_1,\ldots,o_n)}$, resulting in $\dlro$.
We denote the DL \dlro{} without number restrictions as \dlrom{}. Satisfiability checking of
concept expressions w.r.t.~\dlrom{} knowledge bases can be reduced to checking satisfiability of guarded
programs \cite{Heymans+NieuwenborghETAL-OpenAnswProgwith:06}.
\par
A g-hybrid knowledge base consists
of a Description Logic knowledge base and a guarded program.  The
\dllog{} knowledge bases from \cite{Rosa06b} are \emph{weakly safe},
which means that the interaction between the program and the DL knowledge base is
restricted by requiring that  variables which appear in non-DL atoms,
appear in positive non-DL atoms in the body, where DL atoms are atoms involving a concept or role
symbol from the DL knowledge base.
G-hybrid knowledge bases do not require such a restriction; instead,
variables must appear in a \emph{guard} of the rule, but this
guard can be a DL atom as well.  In this paper, we show decidability of g-hybrid
knowledge bases for \dlrom{} knowledge bases by a reduction to
guarded programs, and show that satisfiability
checking of g-hybrid knowledge bases is \dblexptime{}-complete. The DL \dlrom{} is close to
\shoin, the Description Logic underlying OWL DL
\cite{HoPa04b}. Compared with \shoin, \dlrom{} does not include
transitive roles and number restrictions, but does include $n$-ary
roles and complex role expressions.
\par
To see why a combination of rules and ontologies, as proposed in g-hybrid knowledge bases, is useful, and why
the safeness conditions considered so far in the literature are not appropriate in all scenarios, consider the
Description Logic ontology
$$\begin{array}{l}\mathit{FraternityMember}\sqsubseteq \mathit{Drinker}\sqcap\exists \mathit{hasDrinkingBuddy.FraternityMember}
\end{array}$$
which says that fraternity members are drinkers who have drinking buddies, which are also fraternity members.
Now consider the logic program
$$
\begin{array}{lcl} \mathit{problemDrinker(X)}& \leftarrow & \mathit{Drinker(X)},\mathit{not\; socialDrinker(X)}\\
  \mathit{socialDrinker(X)}& \leftarrow & \mathit{Drinker}(X), \mathit{not \; problemDrinker(Y)},\\
  & & \mathit{hasDrinkingBuddy(X,Y)}\\
  \mathit{FraternityMember(John)}& \leftarrow
\end{array}$$
which says that drinkers are by default problem
drinkers, unless it is known that they are social drinkers;
drinkers with drinking buddies who are not problem
drinkers are social drinkers; and John
is a fraternity member. From the combination of the ontology and the
logic program, one would expect to derive that John is a social
drinker, and not a problem drinker. This logic program cannot be
expressed using r-hybrid knowledge bases, or \dllog{}, because the
rules in the program are not weakly safe .  However, the logic
program is \emph{guarded}, and thus part of a valid g-hybrid
knowledge base, which has the expected consequences.

\smallskip

The remainder of the paper starts with an introduction to open answer
set programming and Description Logics in Section \ref{sec:prel}.  Section \ref{sec:ghybr} defines g-hybrid
knowledge bases, translates them to guarded programs when the DL \dlrom{}
is considered, and provides a complexity characterization for
satisfiability checking of these particular g-hybrid knowledge bases.
In Section \ref{sec:related}, we discuss the relation of g-hybrid
knowledge bases with \dllog{} and other related work. We
conclude and give directions for further research in Section
\ref{sec:conclusions}.


\section{Preliminaries}\label{sec:prel}

In this section we introduce Open Answer Set Programming, guarded programs, and the
Description Logic \dlrom.

\subsection{Decidable Open Answer Set Programming}\label{sec:answer}

We introduce the open answer set semantics from
\cite{hvnv-lpnmr2005,Heymans+NieuwenborghETAL-OpenAnswProgwith:06},
modified as in \cite{hvnfv-row2006} such that it does not assume uniqueness of names by default.  \emph{Constants},
\emph{variables}, \emph{terms}, and \emph{atoms} are defined as usual.
A \textit{literal} is an atom $p(\vec{t})$ or a \textit{naf-literal}
$\naf{p(\vec{t})}$, with $\vec{t}$ a tuple of terms.\footnote{We do not allow ``classical'' negation
$\neg$, however, programs with $\neg$ can be reduced to programs
without it, see e.g.  \cite{lifschitz:2001}.} The \textit{positive
part} of a set of literals $\alpha$ is $\posi{\alpha} =
\set{p(\vec{t}) \mid p(\vec{t}) \in \alpha}$ and the \textit{negative
part} of $\alpha$ is $\nega{\alpha} = \set{p(\vec{t}) \mid
\naf{p(\vec{t})} \in \alpha}$.  We assume the existence of the (in)equality predicates $=$ and $\neq$,
usually written in infix notation; $t = s$ is an atom and
$t\neq s$ is short for $\naf{t =s}$.   A \emph{regular} atom is an atom without equality.
For a set $A$ of atoms, $\naf{A} = \set{\naf{l} \mid l \in A}$.
\par
A \textit{program} is a countable set of rules \prule{\alpha}{\beta},
where $\alpha$ and $\beta$ are finite sets of literals,
$\card{\posi{\alpha}} \leq 1$
(but $\nega{\alpha}$ may be of arbitrary
size), and every atom in $\posi{\alpha}$ is regular, i.e.~$\alpha$ contains at most one positive atom, which
may not contain the equality predicate.\footnote{The condition
$\card{\posi{\alpha}} \leq 1$ makes the GL-reduct non-disjunctive,
ensuring that the \emph{immediate consequence operator} is
well-defined, see \cite{Heymans+NieuwenborghETAL-OpenAnswProgwith:06}.} The set $\alpha$ is the
\textit{head} of the rule and represents a disjunction of literals,
while $\beta$ is the \textit{body} and represents a conjunction
of literals.  If $\alpha = \emptyset$, the rule is called a
\textit{constraint}.  \emph{Free rules} are rules of the form
${q(\vec{X})\lor\naf{q(\vec{X})}}\gets$; they enable a choice for the inclusion of atoms in a model.
We call a
predicate $p$  \emph{free} if there is a free rule
${p(\vec{X})\lor\naf{p(\vec{X})}}\gets$.  Atoms, literals, rules, and
programs that do not contain variables are \textit{ground}.
\par
For a literal, rule, or program $o$, let $\cts{o},\vars{o},\preds{o}$ be the constants, variables, and predicates, respectively, in
$o$.  A
\emph{pre-interpretation} $U$ for a program $P$ is a pair $(D,\sigma)$ where $D$
is a non-empty \emph{domain} and $\sigma:\cts{P}\to D$ is a function
which maps all constants in $P$ to elements from $D$.\footnote{In
\cite{Heymans+NieuwenborghETAL-OpenAnswProgwith:06}, we only use the domain $D$ which is there defined as a non-empty superset of
the constants in $P$.  This corresponds to a pre-interpretation $(D,\sigma)$ where
$\sigma$ is the identity function on $D$.} $\ground{P}{U}$ is the
ground program obtained from $P$ by substituting every variable in $P$
with every possible element from $D$ and every constant $c$ with
$\sigma(c)$.  E.g., for a rule $r:p(X)\gets f(X,c)$ and
$U=(\set{x,y},\sigma)$ where $\sigma(c)=x$, we have that the grounding
w.r.t.~$U$ is:
\begin{program}
\tsrule{p(x)}{f(x,x)}
\tsrule{p(y)}{f(y,x)}
\end{program}
\par
Let $\HBase{P}$ be the set of regular atoms obtained from the
language of the ground program $P$.
An \textit{interpretation} $I$ of a ground program  $P$ is a subset of
$\HBase{P}$.  For a ground regular atom $p(\vec{t})$, we write $I
\models p(\vec{t})$ if $p(\vec{t}) \in I$; for an equality atom
$t = s$, we write $I \models t=s$ if $s$ and
$t$ are equal terms.  We write $I \models \naf{p(\vec{t})}$ if $I
\not\models p(\vec{t})$, for $p(\vec{t})$ an atom.  For a set of ground literals $A$, $I\models
A$ holds if $I \models l$ for every $l \in A$.  A ground rule $r:
\prule{\alpha}{\beta}$ is \textit{satisfied} w.r.t.~$I$, denoted $I
\models r$, if $I \models l$ for some $l \in \alpha$ whenever $I
\models \beta$.  A ground constraint \prule{}{\beta} is
satisfied w.r.t.~$I$ if $I \not\models \beta$.
\par
For a ground program $P$ without \NAF, an interpretation $I$ of $P$ is
a \textit{model} of $P$ if $I$ satisfies every rule in $P$; it is an
\textit{answer set} of $P$ if it is a subset minimal model of $P$.
For ground programs $P$ containing \NAF, the \textit{reduct}
\cite{inoue98negation} w.r.t.~$I$ is $P^I$, where $P^I$
consists of \prule{\posi{\alpha}}{\posi{\beta}} for every \prule{\alpha}{\beta}
in $P$ such that $I \models \naf{\nega{\beta}}$ and $I \models \nega{\alpha}$.
$I$ is an \textit{answer set} of $P$ if $I$ is an answer set
of $P^I$.
  Note that allowing negation in the head of rules leads to the
  loss of the \emph{anti-chain property}~\cite{inoue98negation} which states that no answer set
  can be a strict subset of another answer set.  E.g, a rule \prule{a\lor\naf{a}}{} has the answer sets
  $\emptyset$ and $\set{a}$ .  However, negation in the head is required
  to ensure first-order behavior for certain predicates, e.g.,
  when simulating Description Logic reasoning.
\par
In the following, a program is assumed to be a finite set of rules;
infinite programs only appear as byproducts of grounding a finite
program using an infinite pre-interpretation.  An \textit{open interpretation} of
a program  $P$ is a pair $(U, M)$ where $U$ is a pre-interpretation for $P$ and
$M$ is an interpretation of $P_{U}$.  An \textit{open answer set} of
$P$ is an open interpretation $(U,M)$ of $P$ with $M$ an answer set of
$\ground{P}{U}$.  An $n$-ary predicate $p$ in $P$ is
\emph{satisfiable} if there is an open answer set $((D,\sigma),M)$ of
$P$ and a $\vec{x} \in D^n$ such that $p(\vec{x}) \in M$. A program
$P$ is satisfiable iff it has an open answer set. Note that
satisfiability checking of programs can be easily reduced to
satisfiability checking of predicates: $P$ is satisfiable iff $p$ is
satisfiable w.r.t.~$P\cup\set{p(\vec{X})\lor\naf{p(\vec{X})}\gets}$, where $p$ is a predicate symbol not used in $P$ and $\vec{X}$ is a tuple of variables. In the
following, when we speak of satisfiability checking, we refer
to satisfiability checking of predicates, unless specified otherwise.
\par
Satisfiability checking w.r.t.~the open answer set semantics is
undecidable in general.  In
\cite{Heymans+NieuwenborghETAL-OpenAnswProgwith:06}, we identify a
syntactically restricted fragment of programs, so-called
\emph{guarded programs}, for which satisfiability checking is
decidable, which is shown through a reduction to guarded fixed point
logic \cite{gradel}.  The decidability of guarded programs relies on
the presence of a  \emph{guard} in each rule, where a guard is an atom that contains
all variables of the rule.  Formally, a rule
$r:\prule{\alpha}{\beta}$ is \emph{guarded} if there is an atom
$\gamma_b \in {\posi{\beta}}$ such that $\vars{r}\subseteq
\vars{\gamma_b}$; $\gamma_b$ is the \emph{guard} of $r$.  A program
$P$ is a \emph{guarded program (GP)} if every non-free rule in $P$
is {guarded}.  E.g., a rule $a(X,Y)\gets \naf{f(X,Y)}$ is not
guarded, but $a(X,Y)\gets g(X,Y),\naf{f(X,Y)}$ is guarded with guard
$g(X,Y)$.  Satisfiability checking of predicates w.r.t.~guarded
programs is \exptimex{2}-complete
\cite{Heymans+NieuwenborghETAL-OpenAnswProgwith:06} -- a result that
stems from the corresponding complexity in guarded fixed point
logic.

\subsection{The Description Logic \dlrom}\label{subsec:dlr}

\dlr{}\index{\dlr{}} \cite{DL-97,dlbook} is a DL which supports roles of arbitrary arity, whereas
most DLs only support binary roles. We introduce an extension of \dlr{} with nominals, called \dlro{}
\cite{hvnfv-row2006}.  The basic building blocks of \dlro{} are
\emph{concept names} $A$ and \emph{relation names}\index{relation
name} $\mathbf{P}$ where $\mathbf{P}$ denotes an arbitrary $n$-ary
relation for $2\leq n \leq n_\mathit{max}$ and $n_\mathit{max}$ is a
given finite non-negative integer.  Role expressions $\mathbf{R}$ and
concept expressions $C$ are defined as:
\[
\begin{split}
\mathbf{R} &\to \top_n \mid \mathbf{P} \mid \dlrrole\index{\dlrrole{}} \mid \neg \mathbf{R} \mid \mathbf{R_1}\sqcap \mathbf{R_2} \mid \set{(o_1,\ldots,o_n)} \\
C  &\to \top_1 \mid A \mid \neg C \mid C_1 \sqcap C_2 \mid \dlrconcept\index{\dlrconcept{}} \mid\; \dlrconceptless\index{\dlrconceptless{}} \mid \set{o}
\end{split}
\]
where $i$ is between $1$ and $n$ in \dlrrole{};
similarly in \dlrconcept{} and \dlrconceptless{} for $\mathbf{R}$ an
$n$-ary relation.  Moreover, we assume that the above constructs are
 \emph{well-typed}, e.g., $\mathbf{R}_1\sqcap \mathbf{R}_2$ is defined only
for relations of the same arity.  The semantics of \dlro{} is given by
interpretations $\Int=(\DeltaI,\cdot^{\Int})$ where $\DeltaI$ is a
non-empty set, the \emph{domain}, and $\cdot^{\Int}$ is an
interpretation function such that $C^{\Int}\subseteq \DeltaI$,
$\mathbf{R}^{\Int}\subseteq (\DeltaI)^{n}$ for an $n$-ary relation
$\mathbf{R}$, and the following conditions are satisfied
($\mathbf{P}, \mathbf{R}, \mathbf{R}_1$, and $\mathbf{R}_2$ have arity
$n$):
\[
\begin{split}
\top_n^{\Int} &\subseteq (\DeltaI)^{n} \\
\mathbf{P}^{\Int} &\subseteq \top_n^{\Int} \\
(\neg \mathbf{R})^{\Int} &= \setmin{\top_n^{\Int}}{\mathbf{R}^{\Int}} \\
(\mathbf{R}_1\sqcap \mathbf{R}_2)^{\Int} &= \mathbf{R}^{\Int}_1\cap \mathbf{R}^{\Int}_2 \\
\dlrrole^{\Int} &= \set{(d_1,\ldots, d_n) \in \top_n^{\Int} \mid d_i\in C^{\Int}} \\
\end{split}
\]
\[
\begin{split}
\top_1^{\Int} &= \DeltaI \\
A^{\Int} &\subseteq \DeltaI \\
(\neg C)^{\Int} &= \setmin{\DeltaI}{C^{\Int}} \\
({C_1}\sqcap {C_2})^{\Int} &= C_1^{\Int}\cap C_2^{\Int} \\
(\dlrconcept)^{\Int} &= \set{d\in \DeltaI \mid \Exists{ (d_1, \ldots, d_n) \in \mathbf{R}^{\Int}}{d_i=d}} \\
(\dlrconceptless)^{\Int} &= \set{d\in \DeltaI \mid \;\card{\set{ (d_1, \ldots, d_n) \in \mathbf{R}^{\Int}\mid d_i=d}} \leq k} \\
\set{o}^{\Int} &=\set{o^{\Int}}\subseteq \DeltaI\\
\set{(o_1,\ldots, o_n)}^{\Int} &= \set{(o_1^{\Int},\ldots,o_n^{\Int})}
\end{split}
\]
Note that in \dlro{} the negation of role expressions is defined
w.r.t.~$\top_n^{\Int}$ and not w.r.t.~$(\DeltaI)^n$.  A \dlro{} knowledge base $\Sigma$
is a set of terminological axioms and role axioms, which denote subset
relations between concept and role expressions (of the
same arity), respectively.  A terminological axiom $C_1\sqs C_2$ is
\emph{satisfied} by $\Int$ iff $C_1^{\Int} \subseteq C_2^{\Int}$.  A
role axiom $\mathbf{R}_1\sqs \mathbf{R}_2$ is \emph{satisfied} by
$\Int$ iff $\mathbf{R}_1^{\Int}\subseteq \mathbf{R}_2^{\Int}$.  An
interpretation $\Int$ is a \emph{model} of a knowledge base $\Sigma $ (i.e.~$\Sigma $
is satisfied by $\Int$) if all axioms in $\Sigma $ are
satisfied by $\Int $; if $\Sigma $ has a model, then $\Sigma $ is  \emph{satisfiable}.
A concept expression $C$ is satisfiable w.r.t.~a knowledge
base $\Sigma$ if there is a model $\Int$ of $\Sigma$ such that $C^{\Int}\neq
\emptyset$.
\par
Note that for every interpretation $\Int$,
\[
(\set{(o_1,\ldots,o_n)})^{\Int} = (\dlrrolei{1}{\set{o_1}} \sqcap
\ldots \sqcap \dlrrolei{n}{\set{o_n}})^{\Int} \text{.}
\]
Therefore, in the remainder of the paper, we will restrict ourselves to
nominals of the form $\set{o}$.
We denote the fragment of \dlro{} without the number restriction
\dlrconceptless{} with \dlrom{}.


\section{G-hybrid Knowledge Bases}\label{sec:ghybr}

G-hybrid knowledge bases are combinations of Description Logic (DL)
knowledge bases and guarded logic programs (GP).  They are a variant of the r-hybrid knowledge
bases introduced in \cite{rosati:05}.

\begin{definition}
  \label{def:ghybrid}
  Given a Description Logic \dl{}, a \emph{g-hybrid knowledge base}
  is a pair $(\Sigma,P)$ where $\Sigma$ is a \dl{} knowledge base and
  $P$ is a {guarded program}
  .
\end{definition}
Note that in the above definition there are no restrictions on the use of predicate symbols.
We call the atoms
and literals in $P$ that have underlying predicate symbols which correspond to
concept or role names in the DL knowledge base \emph{DL atoms}
and \emph{DL literals}, respectively. {Variables in rules
{are not required } to appear in positive non-DL atoms, which
is the case in, e.g., the \dllog{} knowledge bases in \cite{Rosa06b},
the r-hybrid knowledge bases in \cite{rosati:05}, and the DL-safe rules
in \cite{motik}.}
  DL-atoms can appear in the head of rules, thereby enabling a
  bi-directional flow of information between the DL knowledge base and
  the logic program.
\begin{example}\label{ex:social}
  Consider the \dlrom{} knowledge base $\Sigma$ where
  \lit{socialDrinker} is a concept, \lit{drinks} is a ternary role
  such that, intuitively, $(x,y,z)$ is in the interpretation of
  \lit{drinks} if a person $x$ drinks some drink $z$ with a
   person $y$.  $\Sigma $ consists of the single axiom
  \begin{knowb}
    \taxiom{socialDrinker}{\dlrconcepti{1}{(drinks \sqcap \dlrrolein{3}{3}{\mset{wine}})}\;}
  \end{knowb}
  {
\noindent  which indicates that social drinkers
  drink wine
  with someone.
  Consider a GP $P$ that indicates that someone has an increased risk of alcoholism if
  the person is a social drinker and knows someone from {the association
    of Alcoholics Anonymous (AA).}  Furthermore, we state that
    \emph{john} is a social drinker and knows \emph{michael} from AA:
 \begin{program}
    \tsrule{ problematic(X) } { socialDrinker(X), knowsFromAA(X,Y) }
    \tsrule{ knowsFromAA(john,michael) }{}
    \tsrule{ socialDrinker(john) }{}
  \end{program}
  Together, $\Sigma$ and $P$ form a g-hybrid knowledge
  base. } The literals \lit{socialDrinker(X)} and
  \lit{socialDrinker(john)} are DL atoms where the latter appears in
  the head of a rule in $P$. The literal \emph{knowsFromAA(X,Y)}
  appears only in the program $P$ (and is thus not a DL atom).
\end{example}

Given a DL interpretation
$\Int=(\DeltaI,\cdot^{\Int})$ and a ground program $P$, we define $\Pi(P,\Int)$ as the
\emph{projection} of $P$ with respect to $\Int$, which is obtained as follows:
for every rule $r$ in $P$,

\begin{itemize}
\item if there exists a DL literal in the head of the form
    \begin{itemize}
    \item $A(\vec{t})$ with $\vec{t}\in A^{\Int}$, or
    \item $\naf{A(\vec{t})}$ with $\vec{t}\not\in A^{\Int}$,
    \end{itemize}
then delete $r$,

\item if there exists a DL literal in the body of the form
    \begin{itemize}
    \item $A(\vec{t})$ with $\vec{t}\not\in A^{\Int}$, or
    \item $\naf{A(\vec{t})}$ with $\vec{t}\in A^{\Int}$,
    \end{itemize}
then delete $r$,

\item otherwise, delete all DL literals from $r$.
\end{itemize}
{
Intuitively, the projection ``evaluates'' the program with respect to $\Int $
by removing (evaluating) rules and DL literals consistently with $\Int$;
conceptually this is similar to the reduct, which removes
rules and negative literals consistently with an interpretation of the program.
}
\begin{definition}
  Let $(\Sigma,P)$ be a g-hybrid knowledge base.
  An interpretation of $(\Sigma,P)$ is a tuple $(U,\Int,M)$ such that
\begin{itemize}
\item $U = (D,\sigma)$ is a pre-interpretation for $P$,
\item $\Int=(D,\cdot^\Int)$ is an interpretation of $\Sigma$,
\item $M$ is an interpretation of $\Pi(P_U,\Int)$, and
\item $b^{\Int} = \sigma(b)$ for every constant symbol $b$ appearing both in $\Sigma$ and in $P$.
\end{itemize}
\end{definition}

Then, $(U=(D,\sigma),\Int,M)$ is a  \emph{model} of a g-hybrid knowledge base $(\Sigma,P)$ if
$\Int$ is a model of $\Sigma $ and $M$ is an answer set of $\Pi(P_U,\Int)$.

For $p$ a concept expression from $\Sigma$ or a predicate from $P$, we
say that $p$ is \emph{satisfiable} w.r.t.~$(\Sigma,P)$ if there is a model
$(U,\Int,M)$ such that $p^{\Int}\neq \emptyset$ or $p(\vec{x})\in M$
for some $\vec{x}$ from $D$, respectively.
\begin{example}\label{ex:semanticssocial}
  Consider the g-hybrid knowledge base in Example \ref{ex:social}.
  Take $U = (D,\sigma)$ with $D = \{john,$ $michael,$ $wine,$ $x\}$ and
  $\sigma$ the identity function on the constant symbols in
  $(\Sigma,P)$. Furthermore, define $\cdot^{\Int}$ as follows:
  \begin{itemize}
    \item $\lit{socialDrinker}^{\Int}=\mset{john}$,
    \item $\lit{drinks}^{\Int} = \mset{(john,x,wine)}$,
    \item $\lit{wine}^{\Int} = \lit{wine}$.
  \end{itemize}
  If $M =
  \mset{knowsfromAA(john,michael),problematic(john)}$, then
  $(U,\Int,M)$ is a model of this g-hybrid knowledge base.  Note that the
  projection $\Pi(P,\Int)$ does not contain the rule
  \prule{socialDrinker(john)}{}.
\end{example}

\section{Translation to Guarded Logic Programs}
\label{sec:translation}

In this section we introduce a translation of g-hybrid knowledge bases to guarded
logic programs (GP)  under the open answer set semantics, show that this translation preserves satisfiability,
and use this translation to obtain complexity results for reasoning in g-hybrid
knowledge bases.
Before introducing the translation to guarded programs formally, we
introduce the translation through an example.

Consider the knowledge base in Example
\ref{ex:social}. The axiom
\begin{knowb}
    \taxiom{socialDrinker}{\dlrconcepti{1}{(drinks \sqcap \dlrrolein{3}{3}{\mset{wine}})}}
\end{knowb}
translates to
the constraint
\begin{program}
  \tsrule{}{socialDrinker(X), \naf{(\dlrconcepti{1}{(drinks \sqcap \dlrrolein{3}{3}{\mset{wine}})})(X)}}
\end{program}
Thus, the concept expressions on either side of the $\sqsubseteq$
symbol are associated with a new unary predicate name. For
convenience, we name the new predicates according to the original
concept expressions. The constraint simulates the behavior of the
\dlrom{} axiom. If the left-hand side of the axiom holds and the
right-hand side does not hold, there is a contradiction.
\par
It remains to ensure that those newly introduced predicates behave according
to the DL semantics. First, all the concept and role names occurring in the
axiom above need to be defined as free predicates, in order to
simulate the first-order semantics of concept and role names in DLs.
In DLs, a tuple is either true or false in a given interpretation
(cf.~the law of the excluded middle); this behavior can be captured exactly by the free predicates:
\begin{program}
\tsrule{socialDrinker(X) \lor \naf{socialDrinker(X)}}{}
\tsrule{drinks(X,Y,Z) \lor \naf{drinks(X,Y,Z)}}{}
\end{program}
Note that concept names are translated to
unary free predicates, while
$n$-ary role names are translated to
$n$-ary free predicates.
\par
The definition of the truth symbols $\top_1$ and $\top_3$ which are
implicit in our \dlrom{} axiom (since the axiom contains
a concept name and a ternary role)
are translated to free predicates as well. Note that we do not need a
predicate for $\top_2$ since the axiom does not contain binary
predicates.
\begin{program}
\tsrule{\top_1(X)\lor\naf\top_1(X)}{}
\tsrule{\top_3(X, Y, Z)\lor\naf\top_3(X, Y, Z)}{}
\end{program}
We ensure that, for the ternary \dlrom{} role $\lit{drinks}$, $drinks^{\Int} \subseteq \top_3^{\Int}$
holds by adding the constraint:
\begin{program}
  \tsrule{}{drinks(X, Y, Z), \naf{\top_3(X, Y, Z)}}
\end{program}
To ensure that $\top_1^{\Int} = \DeltaI$, we add the constraint:
\begin{program}
\tsrule{}{\naf{\top_1(X)}}
\end{program}
For rules containing only one variable, we can always assume that
$X = X$ is in the body and acts as the guard of the rule, so that the
latter rule is guarded; cf.~the equivalent rule
\prule{}{\naf{\top_1(X)},X=X}.
\par
We translate the nominal $\mset{wine}$ to the rule
\begin{program}
\tsrule{{\mset{wine}}(wine)}{}
\end{program}
Intuitively, since this rule will be the only rule with the predicate
$\mset{wine}$ in the head, every open answer set of the translated
program will contain $\mset{wine}(x)$ with $\sigma(\lit{wine}) =
x$ if and only if the corresponding interpretation
$\mset{wine}^{\Int} = \set{x}$ for $\lit{wine}^{\Int}=x$.

The \dlrom{} role expression $\dlrrolein{3}{3}{\mset{wine}}$ indicates
the ternary tuples for which the third argument belongs to the
extension of $\{wine\}$, which is translated to the following rule:
\begin{program}
  \tsrule{\dlrrolein{3}{3}{\mset{wine}}(X, Y, Z)}{\top_3(X, Y, Z), \mset{wine}(Z)}
\end{program}
Note that the above rule is guarded by the $\top_3$ literal.
\par
Finally, the concept expression $\lit{(drinks \sqcap \dlrrolein{3}{3}{\mset{wine}})}$ can be
represented by the following rule:
\begin{program}
\ssrule{(drinks \sqcap \dlrrolein{3}{3}{\mset{wine}})(X, Y, Z)} {drinks(X, Y, Z),} \\ && {\dlrrolein{3}{3}{\mset{wine}}(X, Y,Z)}
\end{program}
As we can see, the DL construct $\sqcap$ is translated to
conjunction in the body of a rule.
\par
The \dlrom{} role $\dlrconcepti{1}{(drinks \sqcap
\dlrrolein{3}{3}{\mset{wine}})}$ can be represented using the
following rule:
\begin{program}
  (\dlrconcepti{1}{(drinks \sqcap \dlrrolein{3}{3}{\mset{wine}})})(X)\;\leftarrow\;(drinks \sqcap \dlrrolein{3}{3}{\mset{wine}})(X, Y, Z)
\end{program}
Indeed, the elements which belong to the extension of
\dlrconcepti{1}{(drinks \sqcap \dlrrolein{3}{3}{\mset{wine}})} are
exactly those that are connected to the role
\dlrrolein{3}{3}{\mset{wine}}, as specified in the rule.
\par
This concludes the translation of the DL knowledge base in the
g-hybrid knowledge base of Example \ref{ex:social}.  The program
can be considered as is, since, by definition of g-hybrid knowledge
bases, it is already a guarded program.

\medskip

\noindent We now proceed with the formal translation.  The \emph{closure} \clos{\Sigma} of a
\dlrom{} knowledge base $\Sigma$ is defined as the smallest set satisfying
the following conditions:
\begin{itemize}
\item $\top_1\in \clos{\Sigma}$,
\item for each $C\sqs D$ an axiom in $\Sigma$ (role or terminological), $\set{C,D}\subseteq \clos{\Sigma}$,
\item for every $D$ in \clos{\Sigma}, $\clos{\Sigma}$
contains every subformula which is a concept expression or a role
expression,
\item if \clos{\Sigma} contains an $n$-ary relation name, it
contains $\top_n$.
\end{itemize}
We define
$\Phi(\Sigma)$ as the smallest logic program satisfying the following conditions:
\begin{itemize}
\item For each terminological axiom $C\sqs D\in \Sigma$, $\Phi(\Sigma)$ contains the constraint:
\begin{equation}\label{dlreq:axiom}
\pnorule{}{C(X),\naf{D(X)}}
\end{equation}
\item For each role axiom $\mathbf{R}\sqs \mathbf{S}\in \Sigma$ where
$\mathbf{R}$ and $\mathbf{S}$ are $n$-ary, $\Phi(\Sigma)$ contains:
\begin{equation}\label{dlreq:role}
\pnorule{}{\mathbf{R}(X_1,\ldots,X_n),\naf{\mathbf{S}(X_1,\ldots, X_n)}}
\end{equation}

\item For each $\top_n\in \clos{\Sigma}$, $\Phi(\Sigma)$ contains the free rule:
\begin{equation}\label{dlreq:top}
\pnorule{\top_n(X_1,\ldots, X_n)\lor\naf{\top_n(X_1,\ldots, X_n)}}{}
\end{equation}
Furthermore, for each $n$-ary relation name $\mathbf{P}\in
\clos{\Sigma}$, $\Phi(\Sigma)$ contains:
\begin{equation}\label{dlreq:contop}
\pnorule{}{\mathbf{P}(X_1,\ldots,X_n),\naf{\top_n(X_1,\ldots, X_n)}}
\end{equation}
Intuitively, the latter rule ensures that $\mathbf{P}^{\Int}\subseteq
\top_n^{\Int}$. Additionally, $\Phi(\Sigma)$ has to contain the constraint:
\begin{equation}\label{dlreq:top1}
\pnorule{}{\naf{\top_1(X)}}
\end{equation}
which ensures that, for every element $x$ in the pre-interpretation, $\top_1(x)$
is true in the open answer set. The latter rule ensures that
$\top_1^{\Int} = D$ for the corresponding interpretation. The rule is
implicitly guarded with $X=X$.

\item Next, we distinguish between the types of concept and role expressions
that appear in \clos{\Sigma}.  For each $D\in \clos{\Sigma}$:
\begin{itemize}
\item if $D$ is a concept nominal $\set{o}$, $\Phi(\Sigma)$ contains the fact:
\begin{equation}\label{dlreq:freeconcept}
\pnorule{D(o)}{}
\end{equation}

This fact ensures that $\set{o}(x)$ holds in any open answer set iff $x=\sigma(o) = o^{\Int}$ for an interpretation of $(\Sigma,P)$.

\item if $D$ is a concept name, $\Phi(\Sigma)$ contains:
\begin{equation}
\pnorule{D(X)\lor \naf{D}(X)}{}
\end{equation}
\item if $D$ is an $n$-ary relation name, $\Phi(\Sigma)$ contains:
\begin{equation}\label{dlreq:freerole}
\pnorule{\mathbf{D}(X_1,\ldots, X_n)\lor \naf{\mathbf{D}}(X_1,\ldots, X_n)}{}
\end{equation}
\item if $D=\neg E$ for a concept expression $E$, $\Phi(\Sigma)$ contains the rule:
\begin{equation}\label{dlreq:neg}
\pnorule{D(X)}{\naf{E(X)}}
\end{equation}
Note that we can again assume that such a  rule is guarded by $X=X$.

\item if $D=\neg \mathbf{R}$ for an $n$-ary role expression $\mathbf{R}$, $\Phi(\Sigma)$ contains:
\begin{equation}\label{dlreq:negrole}
{D(X_1,\ldots, X_n)}\gets \top_n(X_1,\ldots, X_n),\naf{\mathbf{R}(X_1,\ldots, X_n)}
\end{equation}
Note that if negation would have been defined w.r.t.~$D^n$
instead of $\top_n^{\Int}$, we would not be able to write the above as a
guarded rule.
\item if $D=E\sqcap F$ for concept expressions $E$ and $F$, $\Phi(\Sigma)$ contains:
\begin{equation}\label{dlreq:conj}
\pnorule{D(X)}{E(X),F(X)}
\end{equation}

\item if $D=\mathbf{E}\sqcap \mathbf{F}$ for $n$-ary role expressions $\mathbf{E}$
and $\mathbf{F}$, $\Phi(\Sigma)$ contains:
\begin{equation}\label{dlreq:conjrole}
\pnorule{D(X_1,\ldots, X_n)}{\mathbf{E}(X_1,\ldots,X_n),\mathbf{F}(X_1,\ldots,X_n)}
\end{equation}
\item if $D=\dlrrole$, $\Phi(\Sigma)$ contains:
\begin{equation}\label{dlreq:dlrrole}
{D(X_1,\ldots,X_i,\ldots, X_n)}\gets\\{\top_n(X_1,\ldots, X_i,\ldots, X_n),C(X_i)}
\end{equation}
\item if $D=\dlrconcept$, $\Phi(\Sigma)$ contains:
\begin{equation}\label{dlreq:exists}
\pnorule{D(X)}{\mathbf{R}(X_1,\ldots, X_{i-1},X,X_{i+1},\ldots, X_n)}
\end{equation}
\end{itemize}
\end{itemize}
The following theorem shows that this translation preserves satisfiability.

\begin{theorem}\label{th:DLtoLPGP}
Let $(\Sigma,P)$ be a g-hybrid knowledge base with $\Sigma$ a
\dlrom{} knowledge base.  Then, a predicate or
concept expression $p$ is satisfiable w.r.t.~$(\Sigma,P)$ iff $p$ is satisfiable w.r.t.~$\Phi(\Sigma)\cup P$.
\end{theorem}
\begin{proof}
($\Rightarrow$) Assume $p$ is satisfiable w.r.t.~$(\Sigma,P)$,
i.e., there exists a model $(U,{\Int},M)$ of $(\Sigma,P)$, with
$U=(D,\sigma)$, in which $p$ has a non-empty extension.  Now, we construct the open
interpretation $(V,N)$ of $\Phi(\Sigma,P)$ as follows.  $V = (D,\sigma')$
with $\sigma': \cts{\Phi(\Sigma)\cup P}\to D$, and
 $\sigma'(x) = \sigma(x)$ for every constant symbol $x$
from $P$ and $\sigma'(x) = x^{\Int}$ for every constant symbol $x$ from
$\Sigma$.  Note that $\sigma'$ is well-defined, since, for a constant
symbol $x$ which occurs in both $\Sigma$ and $P$, we have that
$\sigma(x) = x^{\Int}$. We define the set $N$ as follows:
\begin{multline*}
N=M \cup \set{\pred{C}(x) \mid x \in C^{\Int}, C\in \clos{\Sigma}}\\
\cup \set{\pred{\mathbf{R}}(x_1,\ldots,x_n)\mid
(x_1,\ldots,x_n) \in \mathbf{R}^{\Int}, R\in \clos{\Sigma}}
\end{multline*}
with $C$ and $\mathbf{R}$ concept expressions and role expressions respectively.
\par
It is easy to verify that $(V,N)$ is an open answer set of $\Phi(\Sigma)\cup P$ and $(V,N)$ satisfies $p$.


($\Leftarrow$) Assume $(V,N)$ is an open answer set of
$\Phi(\Sigma)\cup P$ with $V = (D,\sigma')$ such that $p$ is satisfied.  We define the interpretation
$(U,{\Int},N)$ of $(\Sigma, P)$ as follows.
\begin{itemize}
\item $U =(D,\sigma)$
where $\sigma:\cts{P} \to D$ with $\sigma(x) = \sigma'(x)$ (note
that this is possible since $\cts{P} \subseteq
\cts{\Phi(\Sigma)\cup P}$). $U$ is then a pre-interpretation for $P$.

\item  $\Int=(D,\cdot^\Int )$ is defined
such that $A^{\Int} =\set{x \mid \pred{A}(x) \in N}$ for
concept names $A$, $\mathbf{P}^{\Int} = \set{(x_1,\ldots,x_n)
\mid \pred{\mathbf{P}}(x_1,\ldots,x_n) \in N}$ for $n$-ary role names
$\mathbf{P}$ and ${o}^{\Int} = {\sigma'(o)}$, for $o$ a constant symbol
in $\Sigma$ (note that $\sigma'$ is indeed defined on $o$). $\Int$
is then an interpretation of $\Sigma$.

\item $M = \setmin{N}{\set{p(\vec{x})\mid p\in \clos{\Sigma}}}$, such that $M$ is an interpretation of $\Pi(P_U,\Int)$.
\end{itemize}
Moreover, for every constant symbol $b$ which appears in both  $\Sigma$ and  $P$,
$b^{\Int} = \sigma(b)$. As a consequence, $(U,\Int,M)$ is an interpretation of
$(\Sigma,P)$.
\par
It is easy to verify that $(U,\Int,M)$ is a model of $(\Sigma,P)$ which
satisfies $p$.
\end{proof}
\begin{theorem}\label{th:gpispol}
Let $(\Sigma,P)$ be a g-hybrid knowledge base where $\Sigma$ is a
\dlrom{} knowledge base.  Then,
$\Phi(\Sigma)\cup P$ is a guarded program with a size polynomial in the
size of $(\Sigma,P)$.
\end{theorem}
\begin{proof}
  The rules in $\Phi(\Sigma)$ are obviously guarded.  Since $P$ is a guarded program,
  $\Phi(\Sigma)\cup P$ is a guarded program as well.
  \par
  The size of $\clos{\Sigma}$ is of the order $n\log{n}$ where $n$ is
  the size of $\Sigma$. Intuitively, given that the size of an
  expression is $n$, we have that the size of the set of
  its subexpressions is at most the size
  of a tree with depth $\log{n}$ where the size of the subexpressions
  at a certain level of the tree is at most $n$.

The size of
  $\Phi(\Sigma)$ is clearly polynomial in the size of
  $\clos{\Sigma}$, assuming that the arity $n$ of an added role
  expression is polynomial in the size of the maximal arity of role
  expressions in $\Sigma$. If we were to add a relation name
  $\mathbf{R}$ with arity $2^n$, where $n$ is the maximal arity of
  relation names in $C$ and $\Sigma$, the size of $\Sigma$ would
  increase linearly, but the size of $\Phi(\Sigma)\cup P$ would
  increase exponentially: one needs to add, e.g., rules
  \[
  \pnorule{\top_{2^n}(X_1,\ldots, X_{2^n})\lor
    \naf{\top_{2^n}(X_1,\ldots, X_{2^n})}}{}
  \]
    which introduce an exponential number of arguments while the size of the
    role $\mathbf{R}$ does not depend on its arity.
\end{proof}

Note that in g-hybrid knowledge bases, we consider \dlrom{}, which is \dlro{}
without expressions of the form
$\dlrconceptless$, since such expressions cannot be simulated with guarded
programs.
E.g., consider the concept expression $\leq 1 [\$ 1] R$ where $R$ is a
binary role.  One can simulate the $\leq$ by
negation as failure:
\[
\pnorule{\leq 1 [\$ 1] R(X)}{\naf{q(X)}}
\]
for some new $q$,  with $q$ defined such that there are at least 2
different $R$-successors:
\[
\pnorule{q(X)}{R(X,Y_1),R(X,Y_2), Y_1\neq Y_2}
\]
However, the latter rule is not guarded -- there is no atom
that contains $X$, $Y_1$, and $Y_2$.  So, in general, expressing number
restrictions such as \dlrconceptless{} is out of reach for
GPs. From Theorems \ref{th:DLtoLPGP} and
\ref{th:gpispol} we obtain the following corollary.
\begin{corollary}\label{cor:exp}
Satisfiability checking w.r.t.~g-hybrid knowledge
bases $(\Sigma,P)$, with $\Sigma$ a
\dlrom{} knowledge base, can be polynomially reduced to satisfiability checking w.r.t.
GPs.
\end{corollary}

Since satisfiability checking w.r.t.~GPs is \dblexptime-complete
\cite{Heymans+NieuwenborghETAL-OpenAnswProgwith:06}, we obtain the same \dblexptime{} characterization for
g-hybrid knowledge bases. We first make explicit a corollary of
Theorem \ref{th:DLtoLPGP}.
\begin{corollary}\label{lem:red}
Let $P$ be a guarded program.  Then, a concept or role expression $p$ is
satisfiable w.r.t.~$P$ iff $p$ is satisfiable w.r.t.~$(\emptyset,P)$.
\end{corollary}
\begin{theorem}\label{th:complete}
Satisfiability checking w.r.t.~g-hybrid knowledge
bases where the DL part is a
\dlrom{} knowledge base is \dblexptime-complete.
\end{theorem}
\begin{proof}
Membership in \dblexptime{} follows from Corollary \ref{cor:exp}.
Hardness follows from \dblexptime-hardness of satisfiability checking
w.r.t.~GPs and the reduction to satisfiability checking in
Corollary \ref{lem:red}.
\end{proof}

\section{Relation with \dllog{} and other Related Work}\label{sec:related}

In \cite{Rosa06b}, so-called \dllog{} knowledge bases combine a
Description Logic knowledge base with a \emph{weakly-safe} disjunctive
logic program.
Formally, for a particular Description Logic \dl{}, a \emph{\dllog{} knowledge base}
is a pair $(\Sigma,P)$ where $\Sigma$ is a $\dl$
knowledge base consisting of a \emph{TBox} (a set of terminological
axioms) and an \emph{ABox} (a set of \emph{assertional axioms}), and $P$ contains rules
$\prule{\alpha}{\beta}$  such that for every rule $r:\prule{\alpha}{\beta} \in P$:
\begin{itemize}
  \item $\nega{\alpha} = \emptyset$,
  \item  $\nega{\beta}$ does not contain DL atoms (\emph{DL-positiveness}),
  \item each variable in $r$ occurs in $\posi{\beta}$ (\emph{Datalog safeness}), and
  \item each variable in $r$ which occurs in a non-DL atom, occurs in a non-DL atom
    in $\posi{\beta}$ (\emph{weak safeness}).
\end{itemize}
\par
The semantics for \dllog{} is the same as that of g-hybrid knowledge
bases\footnote{Strictly speaking, we did not define answer sets of
disjunctive programs, however, the definitions of Subsection
\ref{sec:answer} can serve for disjunctive programs without
modification. Also, we did not consider ABoxes in our definition of
DLs in Subsection \ref{subsec:dlr}. However, the extension of the
semantics to DL knowledge bases with ABoxes is straightforward.},
with the following exceptions:
\begin{itemize}
  \item We do not require the
    \emph{standard name assumption}, which basically says that the domain of every
    interpretation is essentially the same infinitely countable set of
    constants. Neither do we have the implied \emph{unique name
    assumption}, making the semantics for  g-hybrid knowledge
    bases more in line with current Semantic Web standards such as OWL
    \cite{owl} where neither the standard names assumption nor the unique
    names assumption applies.  Note that Rosati also presented a version
    of hybrid knowledge bases which does not adhere to the unique name
    assumption in an
    earlier work \cite{Rosa05b}. However, the grounding of the program part is with
    the constant symbols explicitly appearing in
    the program or DL part only, which yields a less tight integration
    of the program and the DL part than in \cite{Rosa06b} or in
    g-hybrid knowledge bases.
  \item We define an interpretation as a triple
    $(U,\Int,M)$ instead of a pair $(U,\Int')$ where $\Int' = \Int \cup
    M$; this is, however, equivalent to \dllog.
\end{itemize}
The key
differences of the two approaches are:
\begin{itemize}
  \item The programs considered in \dllog{} may have multiple positive literals in the head,
whereas we allow at most one.  However, we allow negative literals in the head,
whereas this is not allowed in \dllog{}.  Additionally, since DL-atoms are interpreted classically,
we may simulate positive DL-atoms in the head through negative DL-atoms in the body.
  \item Instead of Datalog safeness we require \emph{guardedness}.
    Whereas with Datalog safeness every variable in the rule should
    appear in some positive atom of the body of the rule, guardedness
    requires that there is a positive atom that contains every
    variable in the rule, with the exception of free rules. E.g., \prule{a(X)}{b(X),c(Y)} is Datalog
    safe since $X$ appears in $b(X)$ and $Y$ appears in $c(Y)$, but this rule
    is not guarded since there is no atom that contains both $X$ and
    $Y$.  Note that we could easily extend the
    approach taken in this paper to \emph{loosely guarded programs}
    which require that every two variables in the rule should appear
    together in a positive atom, However, this would still be less
    expressive than Datalog safeness.
  \item We do not have the requirement for weak safeness, i.e., head
    variables do not need to appear positively in a non-DL atom. The
    guardedness may be provided by a DL atom.
    \begin{example}
      Example \ref{ex:social} contains the rule
      \begin{program}
      \tsrule{ problematic(X) } { socialDrinker(X), knowsFromAA(X,Y) }
      \end{program}
      This allows to deduce that $X$ might be a problem case
      even if $X$ knows someone from the AA but is not drinking with
      that person. Indeed, as illustrated by the model in
      Example \ref{ex:social}, \lit{john} is drinking wine with some
      anonymous $x$ and knows \lit{michael} from the AA. More correct
      would be the rule
      \begin{program}
    \tsrule{problematic(X,Z)}{drinks(X,Y,Z),knowsFromAA(X,Y)}
      \end{program}
      where we explicitly say that $X$ and $Y$ in the \lit{drinks} and
      \lit{knowsFromAA} relations should be the same, and we extend the
      \lit{problematic} predicate with the kind of drink that $X$ has
      a problem with.  Then, the head
      variable $Z$ is guarded by the DL atom \lit{drinks} and the rule
      is thus not weakly-safe, but is guarded nonetheless. Thus, the
      resulting knowledge base {is} not a \dllog{} knowledge base,
      but {is} a g-hybrid knowledge base.
    \end{example}
  \item We do not have the requirement for DL-positiveness, i.e., DL
    atoms may appear negated in the body of rules (and also in the
    heads of rules).  However, one could allow this in \dllog{}
    knowledge bases as well, since $\naf{A(\vec{X})}$ in the body of
    the rule has the same effect as $A(\vec{X})$ in the head,
    {
    where} the
    latter is allowed in \cite{Rosa06b}. Vice versa, we can also
    loosen our restriction on the occurrence of positive atoms in the
    head (which allows at most one positive atom in the head), to allow
    {
    for an arbitrary number of positive DL atoms in the head (but
    still keep the number of positive non-DL atoms limited to one).
    E.g., a rule $\prule{p(X)\lor A(X)}{\beta}$, where $A(X)$ is a DL
    atom, is not a valid rule in
    the programs we considered since the head contains more than one
    positive atom.  However, we can always rewrite such a rule to
    $\prule{p(X)}{\beta, \naf{A(X)}}$, which contains at most one
    positive atom in the head.
    }

{Arguably, DL atoms should not be allowed to occur negatively, because
DL predicates are interpreted classically and thus the negation in
front of the DL atom is not nonmonotonic. However, Datalog predicates
which depend on DL predicates are also (partially) interpreted
classically, and DL atoms occurring negatively in the body
are equivalent to DL atoms occurring positively in the head which allows us to partly overcome
our limitation of rule heads to one positive atom. }
  \item We do not take into account ABoxes in the DL knowledge base.  However, the DL we consider includes
    nominals such that one can simulate the ABox using terminological
    axioms.  Moreover, even if the DL does not include nominals, the
    ABox can be written as ground facts in a program and ground facts
    are always guarded.
  \item Decidability for satisfiability
    checking\footnote{\cite{Rosa06b} considers checking satisfiability
    of knowledge bases rather than satisfiability of predicates.
    However, the former can easily be reduced to the latter.}
    of \dllog{} knowledge bases is guaranteed if
    decidability of the
    conjunctive query containment problem is guaranteed for the DL at hand.
    In contrast, we relied on a translation of
    DLs to guarded programs for establishing decidability, and, as explained in the previous
    section, not all DLs (e.g.~those with number restrictions) can be translated
    to such a GP.
\end{itemize}

\smallskip

We briefly mention \al-log \cite{doni-lenz-nard-scha-98}, which is a predecessor of
\dllog. \al-log considers \alc{} knowledge bases for the DL part and a set of positive Horn clauses
for the program part.
Every variable must appear in a positive atom in the body, and
concept names are the only DL predicates which may be used in the rules,
and they may only be used in rule bodies.

\cite{hustadt2003}
and \cite{swift04} simulate reasoning in DLs with an LP formalism by
using an intermediate translation to first-order clauses.  In
\cite{hustadt2003}, $\mathcal{SHIQ}$ knowledge bases are reduced to
first-order formulas, to which the basic superposition
calculus\index{basic superposition calculus} is applied.
\cite{swift04} translates $\mathcal{ALCQI}$ concept expressions to
first-order formulas, grounds them with a finite number of constants,
and transforms the result to a logic program.  One can use a finite
number of constants by the finite model property of $\mathcal{ALCQI}$.
In the presence of terminological axioms this is no longer possible
since the finite model property is not guaranteed to hold.
\par
In \cite{levy96carin}, the DL \alcnr{} ($\mathcal{R}$ stands for role
intersection) is extended with Horn clauses
${q(\vec{Y})}\gets{p_1(\vec{X}_1),\ldots, p_n(\vec{X}_n)}$ where the
variables in $\vec{Y}$ must appear in
$\vec{X}_1\cup\ldots\cup\vec{X}_n$; $p_1, \ldots, p_n$ are either
concept or role names, or ordinary predicates not appearing in the
DL part, and $q$ is an ordinary predicate.   There is no safeness in
the sense that every variable must appear in a non-DL atom.  The
semantics is defined through extended interpretations that
satisfy both the DL and clauses part (as FOL formulas).
Query answering is undecidable if recursive Horn clauses are allowed,
but decidability can be regained by restricting the DL part or by
enforcing that the clauses are role safe (each variable in a role atom
$R(X,Y)$ for a role $R$ must appear in a non-DL atom).  Note that the
latter restriction is less strict than the DL-safeness\footnote{DL-safeness is
a restriction of the earlier mentioned weak safeness.} of
\cite{motik}, where also variables in concept atoms $A(X)$ need to
appear in non-DL atoms.  On the other hand, \cite{motik} allows for
the more expressive DL \shoind, and the head predicates may be DL
atoms as well. Finally, SWRL~\cite{Horrocks+Patel-Schneider-PropRuleLang:04}
can be seen as an extension of~\cite{motik} without any safeness restriction,
which results in the loss of decidability of the formalism.
Compared to our work, we consider a slightly less expressive Description Logic,
but we consider logic programs with nonmonotonic negation, and require guardedness,
rather than role- or DL-safeness, to guarantee decidability.
\par
In \cite{eiter2004} \emph{Description Logic
programs}\index{description logic programs} are introduced; atoms in
the program component may be \emph{dl-atoms} with which one can query
the knowledge in the DL component.  Such \emph{dl-atoms} may specify information
from the logic program which needs to be taken into account when evaluating the query,
yielding a bi-directional
flow of information.  This leads to a minimal interface between the DL knowledge base and the logic program,
enabling a very loose integration, based on an entailment relation.  In contrast,
we propose a much tighter integration between the rules and the ontology,
with interaction based on single models
rather than entailment. For a detailed discussion of these two kinds of interaction, we refer to
\cite{bruijn06-repres-issues-about-combin-of}.

Two recent approaches~\cite{motik07,bruijn07} use an embedding in a nonmonotonic modal logic
for integrating nonmonotonic logic programs and ontologies based on classical logic (e.g.~DL).
\cite{motik07} use the nonmonotonic logic of Minimal Knowledge and Negation as Failure (MKNF) for the
combination, and show decidability of reasoning in case reasoning in the considered description logic is decidable, and the DL safeness condition~\cite{motik} holds for the rules
in the logic program.
In our approach, we do not require  such a safeness condition, but require the rules to be
\emph{guarded}, and make a semantic distinction between DL predicates and rule predicates.
\cite{bruijn07} introduce several embeddings of non-ground logic programs
in first-order autoepistemic logic (FO-AEL), and compare them under combination with
classical theories (ontologies). However, \cite{bruijn07}
do not address the issue of decidability or reasoning of such combinations.

Finally,~\cite{bruijn-et-al-06} use Quantified Equilibrium Logic as a single unifying language
to capture different approaches to hybrid knowledge bases, including
the approach presented in this paper.  Although
we have presented a translation of g-hybrid knowledge bases to guarded logic programs,
our direct semantics is still based on two modules, relying on separate interpretations for
the DL knowledge base and the logic program,
whereas~\cite{bruijn-et-al-06} define equilibrium models, which serve
to give a unifying semantics to the hybrid knowledge base. The approach of
\cite{bruijn-et-al-06} may be used to define a notion of equivalence
between, and may lead to new algorithms for reasoning with, g-hybrid knowledge bases.

\section{Conclusions and Directions for Further Research}\label{sec:conclusions}

We defined g-hybrid knowledge bases which combine Description Logic
(DL) knowledge bases with guarded logic
programs. In particular, we combined knowledge bases of the DL
\dlrom{}, which is close to OWL DL, with
guarded programs, and showed decidability of this framework by a
reduction to guarded programs under the open answer set semantics
\cite{hvnv-lpnmr2005,Heymans+NieuwenborghETAL-OpenAnswProgwith:06}.
 We discussed the relation with
\dllog{} knowledge bases: g-hybrid knowledge bases overcome some of
the limitations of \dllog{}, such as the unique names
assumption, Datalog safeness, and
weak DL-safeness, but introduce the requirement of guardedness.  At present,
a significant disadvantage of our approach is the lack of support for
DLs with number restrictions which is inherent to the use of guarded
programs as our decidability vehicle.  A solution for this would be to
consider other types of programs, such as \emph{conceptual logic
programs} \cite{hvnv-amaijournal2006}. This would allow for the
definition of a hybrid knowledge base $(\Sigma,P)$ where  $\Sigma$ is
a \shiq{} knowledge base and $P$ is a conceptual logic program since
\shiq{} knowledge bases can be translated to conceptual logic
programs.
\par
Although there are known complexity bounds for several fragments of open answer set
programming (OASP), including the guarded fragment considered in this paper, there are no known
effective algorithms for OASP.  Additionally, at presence, there are no implemented
systems for open answer set
programming.  These are part of future work.

%

\bibliographystyle{acmtrans}
\bibliography{tplp_alpsws2006}
\end{document}